\pdfoutput=1
\RequirePackage{ifpdf}
\ifpdf 
\documentclass[pdftex]{sigma}
\else
\documentclass{sigma}
\fi

\usepackage{subfigure}

\numberwithin{equation}{section}

\newtheorem{Theorem}{Theorem}[section]
\newtheorem*{Theorem*}{Theorem}

\newtheorem{Lemma}[Theorem]{Lemma}

 { \theoremstyle{definition}

 }

\begin{document}

\allowdisplaybreaks

\newcommand{\arXivNumber}{2404.18372}

\renewcommand{\PaperNumber}{091}

\FirstPageHeading

\ShortArticleName{Integrable Semi-Discretization for a Modified Camassa--Holm Equation}

\ArticleName{Integrable Semi-Discretization for a Modified\\ Camassa--Holm Equation with Cubic Nonlinearity}

\Author{Bao-Feng FENG~$^{\rm a}$, Heng-Chun HU~$^{\rm b}$, Han-Han SHENG~$^{\rm cd}$, Wei YIN~$^{\rm ae}$ and Guo-Fu YU~$^{\rm d}$}

\AuthorNameForHeading{B.-F.~Feng, H.-C.~Hu, H.-H.~Sheng, W.~Yin and G.-F.~Yu}

\Address{$^{\rm a)}$~School of Mathematical and Statistical Sciences,\\
\hphantom{$^{\rm a)}$}~The University of Texas Rio Grande Valley, Edinburg, Texas 78541, USA}
\EmailD{\href{mailto:baofeng.feng@utrgv.edu}{baofeng.feng@utrgv.edu}}

\Address{$^{\rm b)}$~College of Science, University of Shanghai for Science and Technology,\\
\hphantom{$^{\rm b)}$}~Shanghai 200093, P.R.~China}
\EmailD{\href{mailto:hhengchun@163.com}{hhengchun@163.com}}

\Address{$^{\rm c)}$~Department of Mathematics, City University of Hong Kong,\\
\hphantom{$^{\rm c)}$}~Tat Chee Avenue, Kowloon, Hong Kong, P.R.~China}
\EmailD{\href{mailto:tutu123@sjtu.edu.cn}{tutu123@sjtu.edu.cn}}

\Address{$^{\rm d)}$~School of Mathematical Sciences, CMA-Shanghai, Shanghai Jiao Tong University, \\
\hphantom{$^{\rm d)}$}~Shanghai 200240, P.R.~China }
\EmailD{ \href{mailto:gfyu@sjtu.edu.cn}{gfyu@sjtu.edu.cn}}

\Address{$^{\rm e)}$~Department of Mathematics, South Texas College, McAllen, Texas 78501, USA}
\EmailD{\href{mailto:wei.yin01@utrgv.edu}{wei.yin01@utrgv.edu}}

\ArticleDates{Received April 30, 2024, in final form October 07, 2024; Published online October 12, 2024}

\Abstract{In the present paper, an integrable semi-discretization of the modified Camassa--Holm (mCH) equation with cubic nonlinearity is presented. The key points of the construction are based on the discrete Kadomtsev--Petviashvili (KP) equation and appropriate definition of discrete reciprocal transformations. First, we demonstrate that these bilinear equations and their determinant solutions can be derived from the discrete KP equation through Miwa transformation and some reductions. Then, by scrutinizing the reduction process, we obtain a set of semi-discrete bilinear equations and their general soliton solutions in the Gram-type determinant form. Finally, we obtain an integrable semi-discrete analog of the mCH equation by introducing dependent variables and discrete reciprocal transformation. It is also shown that the	semi-discrete mCH equation converges to the continuous one in the continuum limit.}


\Classification{35Q53; 37K10; 35C05; 37K40}

\section{Introduction}
	In this paper, we are concerned with integrable discretization of the following modified Camassa--Holm (mCH) equation with cubic nonlinearity
\begin{align}
m_t+\big[m\big(u^2-u_x^2\big)\big]_x=0,\qquad m=u-u_{xx}.\label{mch}
\end{align}
Here $u=u(x,t)$ is a real valued function of time $t$ and a spatial variable $x$, and the subscripts $x$ and $t$ appended to $m$ and $u$ denote partial differentiation. It was firstly proposed by Fuchssteiner and Fokas in 1981 (see \cite[equation~(32)]{Fuchssteiner-Fokas}) as a special case of a more general system.	Then it appeared in the papers of Fokas \cite{Fokas}, Fuchssteiner \cite{Fuch}, Olver and Rosenau \cite{olv}, and later was rediscovered by Qiao \cite{qiao,qiao1}.
The mCH equation \eqref{mch} has attracted considerable attention over the past two decades due to its rich mathematical structure and solutions. It has been extensively investigated in various areas, including well-posedness, regularization, the Cauchy problem, the Riemann--Hilbert problem, long-time asymptotics, and the Liouville correspondence with the modified Korteweg--de Vries (KdV) equation \cite{RHP,Qu1,regularization,FORQ1,OlverQu2016,cauchy,Qu3,Tang,longtime,YF2022}. Matsuno presented a compact parametric representation of the smooth bright multisoliton solutions for the mCH equation via the Hirota's bilinear method \cite{matsuno1}, while Hu et al.\ derived its Gram-type determinant solution from the extended Kadomtsev--Petviashvili (KP) hierarchy with negative flow \cite{hu-soliton}. Several groups also constructed the smooth soliton solutions through Darboux transformation/B\"acklund transformation method \cite{FORQ2, BacklundQP, DarbouxQiao} and Lie algebraic approach \cite{Lie}. In \cite{wave-break}, the wave-breaking problem and the existence of single and multi-peakon solutions to the mCH equation have been discussed. Recently, Chang et al.\ have investigated the Lax integrability and the conservative peakon solutions in a series of work \cite{Chang1,Chang2,Chang3}. Gao et al.\ studied the patched peakon weak solution \cite{Patched}, and the conservative sticky peakons \cite{sticky-peakon}. Other related problem such as blow-up phenomena and the stability including the orbital stability have been studied by several authors \cite{oscillation,stab,orbital,blow-up}.

Recently, research on discrete integrable systems has garnered significant attention due to its connections to several other fields, including random matrices, quantum field theory, numerical algorithms, orthogonal and biorthogonal polynomials, and random matrices \cite{Disbook}. There are far fewer instances of discrete integrable systems and analytical tools available as compared to continuous integrable systems. On the other hand, discrete integrable systems are seen to be more basic and universal than continuous ones \cite{Fu-Nijhoff}. The authors have conducted extensive research in finding integrable discretizations of soliton equations, including the short pulse equation~\cite{Feng1,Feng2}, (2+1)-dimensional Zakharov equation \cite{Yu}, the Camassa--Holm (CH) equation~\cite{d-CH3,Feng3}, the Degasperis--Procesi (DP) equaiton \cite{dDP}, the generalized sine-Gordon equation~\cite{gsg1,gsg2} and the mCH equation with cubic nonlinearity and linear
dispersion term \cite{mch-SYF} via Hirota's bilinear method.

It should be commented that there exists a mCH equation with cubic nonlinearity and linear dispersion term
\begin{align}
m_t+\big[m\big(u^2-u_x^2\big)\big]_x+2\kappa^2 u_x=0,\qquad m=u-u_{xx},\label{mch1}
\end{align}
whose bilinear equations are totally different from those of equation~\eqref{mch}. The mCH equation with linear dispersion term were derived in \cite{matsuno2} and also in \cite{mch-SYF} as the reduction of the negative flow of the deformed KdV hierarchy. Although in \cite{mch-SYF} we have proposed an integrable semi-discretization of the mCH equation with linear dispersion term, i.e., equation~\eqref{mch1}, to the best of our knowledge, integrable discrete analogues of equation~\eqref{mch} (the mCH equation without linear dispersion term) have not been reported yet. There are mainly two challenging points in the construction. Firstly, bilinear equations of the mCH equation \eqref{mch} are reduced from the extended KP hierarchy with negative flow. The non-original location of one of the poles presents a challenge in constructing its discrete analogue. Secondly, as shown in Section \ref{discretization}, we have to define a second discrete counterpart for the same continuous variable in order to obtain an explicit form of the semi-discrete mCH equation. Hence, it is a natural but definitely not a~trivial problem to generate a semi-discrete version for the mCH equation \eqref{mch}.

In this paper, upon introducing appropriate Miwa transformation, we derive successfully the two sets bilinear mCH equation from the discrete KP equation. As a byproduct, integrable semi-discrete bilinear mCH equation and the corresponding Gram-type determinant solutions are obtained. Under the discrete reciprocal transformation and dependent variable transformation, an integrable semi-discrete analog of the mCH equation is given.

The outline of the paper is as follows. In Section \ref{structure}, we review the bilinear forms and
determinant solutions of the mCH equation, which can be reduced from the discrete KP equation and its $\tau$-function through a series of transformations including Miwa transformation. In Section \ref{discretization}, by scrutinizing the process in deriving the bilinear mCH equation from the discrete KP equation, we propose semi-discrete analogues of bilinear mCH equations. Based on these discrete bilinear equations, we construct an integrable semi-discrete mCH equation and present its $N$-soliton solutions. Section \ref{conclusion} is devoted to a brief summary and discussion.

\section[From the discrete KP equation to the modified Camassa--Holm equation]{From the discrete KP equation\\ to the modified Camassa--Holm equation}\label{structure}
In this section, we first review the results in \cite{hu-soliton} about the bilinear form of the mCH equation.
The mCH equation \eqref{mch} can be transformed into the following bilinear equations
\begin{align}
& \big(2 D_{{s}} D_y^2+2 D_{ {s}} D_y-4 D_y\big) g \cdot f=0, \label{BL1}\\
& \big(D_y^2+D_y\big) g \cdot f=0, \label{BL2}
\end{align}
through the reciprocal transformation
$
 x=y+ {s}+2 \ln \frac{g}{f}$, $t= {s}
$,
and the dependent variable transformation
$
u=1-(\ln f g)_{y {s}}$,
where $D_x$ is the Hirota D-operator defined by
\begin{align*}
D_x^n f \cdot g=\left(\frac{\partial}{\partial x}-\frac{\partial}{\partial y}\right)^n f(x) g(y) |_{y=x}.
\end{align*}
Next, we give a lemma regarding bilinear equations of the mCH equation \eqref{mch} and show the correspongding reductions.

\begin{Lemma}
	The following bilinear equations
	\begin{align}
	& \big(D_{x_1}^2-D_{x_2}+2 c D_{x_1}\big) \tau_n \cdot \tau_{n+1}=0,\label{mkp1}\\
	& \big(D_{x_{-1}}\big(D_{x_1}^2-D_{x_2}+2 c D_{x_1}\big)-4 D_{x_1}\big) \tau_n \cdot \tau_{n+1}=0, \label{mkp2}
	\end{align}
	admit the Gram-type determinant solutions
\smash{$
	\tau_n=\det\limits_{1\leqslant i,j \leqslant N} \big(m_{ij}^{(n)}\big)$},
	where the matrix elements are defined as
	\begin{align*}
	& m_{i j}^{(n)}=c_{i j}+\frac{1}{p_i+q_j}\left(-\frac{p_i-c}{q_j+c}\right)^{-n} {\rm e}^{\xi_i+\eta_j}, \qquad \xi_i=p_i x_1+p_i^2 x_2+\frac{1}{p_i-c} x_{-1}+\xi_{i 0}, \\
	& \eta_j=q_j x_1-q_j^2 x_2+\frac{1}{q_j+c} x_{-1}+\eta_{j 0},
	\end{align*}	
	and $c_{ij}$, $p_i$, $q_j$, $\xi_{i0}$, $\eta_{j0}$, $c$ are constants.
\end{Lemma}
\begin{proof}
	The discrete Kadomtsev--Petviashvili (dKP) equation, or the Hirota--Miwa (HM) equation,
	\begin{gather}
	 a_1(a_2-a_3) \tau(k_1+1, k_2, k_3) \tau(k_1, k_2+1, k_3+1)\nonumber \\
	\qquad {}+a_2(a_3-a_1) \tau(k_1, k_2+1, k_3) \tau(k_1+1, k_2, k_3+1)\nonumber \\
	\qquad {}+a_3(a_1-a_2) \tau(k_1, k_2, k_3+1) \tau(k_1+1, k_2+1, k_3)=0,\label{HM}
	\end{gather}
	was proposed independently by Hirota \cite{Hirota-1981} and Miwa \cite{Miwa-1982} in early 1980s. It is known that the discrete KP equation admits a general solution in terms of the following Gram-type determinant~\cite{OHTI_JPSJ}
		\begin{align}	\tau(k_1,k_2,k_3,k_4)=\left|m_{ij}\right|=\left|c_{ij}+\frac{1}{p_i+q_j} \left(-\frac{p_i}{q_j}\right)^{-k_4}\prod_{l=1}^3\left(\frac{1-a_l p_i}{1+a_l q_j}\right)^{-k_l}\right|_{1\leqslant i,j \leqslant N}.\label{gram}
		\end{align}
		Here $k_4$ remains unchanged in each of the tau-function so it is hidden in \eqref{HM}.	Notice that the element in Gram-type solution \eqref{gram} of the discrete KP equation \eqref{HM} can be rewritten as
		\begin{align*}
		m_{ij}={}
		& c_{ij}+\frac{1}{p_i+q_j} \left(-\frac{p_i}{q_j}\right)^{-k_4} \left(\frac{1-a_1 p_i}{1+a_1 q_j}\right)^{-k_1}\left(\frac{1-a_2 p_i}{1+a_2 q_j}\right)^{-k_2}\left(\frac{1-a_3 p_i}{1+a_3 q_j}\right)^{-k_3} \\
		={}& c_{ij}+\frac{1}{p_i+q_j}\left(-\frac{p_i}{q_j}\right)^{-(k_3+k_4)}\left(\frac{1-a_1 p_i}{1+a_1 q_j}\right)^{-k_1}\left(\frac{1-a_2 p_i}{1+a_2 q_j}\right)^{-k_2}\left(\frac{1-a_3^{-1} p_i^{-1}}{1+a_3^{-1} q_j^{-1}}\right)^{-k_3} \\
		={}& c_{ij}+\frac{1}{\tilde{p}_i+\tilde{q}_j}\left(-\frac{\tilde{p}_i-c}{\tilde{q}_j+c}\right)^{-(k_3+k_4)}\left(\frac{1-b_1 \tilde{p}_i}{1+b_1 \tilde{q}_j}\right)^{-k_1}\left(\frac{1-b_2 \tilde{p}_i}{1+b_2 \tilde{q}_j}\right)^{-k_2}\left(\frac{1-d p_i^{-1}}{1+d q_j^{-1}}\right)^{-k_3},
		\end{align*}
		where $\tilde{p}_i=p_i+c$, $\tilde{q}_i=q_i-c$, $b_1^{-1}=a_1^{-1}+c$, $b_2^{-1}=a_2^{-1}+c$, $d=a_3^{-1}$. 
		Redefine $n=k_3+k_4$, we arrive at a 4 dimension equation.
		In solution, we let $\tilde{p}_i \to p_i$, $\tilde{q}_j \to q_j$.
		Thus, the dKP equation~\eqref{HM} becomes the discrete deformed modified KP equation
		\begin{gather}
		 \big(d-b_2^{-1}+c\big) \tau_n(k_1+1, k_2, { {k_3}}) \tau_{n+1}(k_1, k_2+1, { {k_3}}+1)\nonumber \\
		\qquad {}+\big(b_1^{-1}-c-d\big) \tau_n(k_1, k_2+1, { {k_3}}) \tau_{n+1}(k_1+1, k_2, { {k_3}}+1)\nonumber \\
		\qquad {}+\big(b_2^{-1}-b_1^{-1}\big) \tau_{n+1}(k_1, k_2, { {k_3}}+1) \tau_n(k_1+1, k_2+1, { {k_3}})=0,\label{dmkp}
		\end{gather}
	 which admits the Gram-type determinant solutions
\smash{$
		\tau_n(k_1,k_2,k_3)=\left|m_{ij}\right|_{1\leqslant i,j\leqslant N}$},
		where
		\begin{align*}
		m_{ij}=c_{ij}+\frac{1}{{p}_i+{q}_j}\left(-\frac{{p}_i-c}{{q}_j+c}\right)^{-n}\left(\frac{1-b_1 {p}_i}{1+b_1 {q}_j}\right)^{-k_1}\left(\frac{1-b_2 {p}_i}{1+b_2 {q}_j}\right)^{-k_2}\left(\frac{1-d (p_i-c)^{-1}}{1+d (q_j+c)^{-1}}\right)^{-k_3}.
		\end{align*}
	Applying Miwa transformation
	\begin{align*}
	&x_1=\sum\limits_{j=1}^2 k_jb_j,\qquad x_2=\frac{1}{2}\sum\limits_{j=1}^2k_jb_j^2,\qquad \ldots, \qquad x_k=\frac{1}{k}\sum_{j=1}^2k_jb_j^k,\\
	&x_{-1}={ {k_3}}d,\qquad { x_{-2}}=\frac{1}{2} { {k_3}}d^2,\qquad \ldots,\qquad x_{-k}=\frac{1}{k} { {k_3}}d^k,\qquad {\ k=1,2,\ldots,}
	\end{align*}
	we obtain an infinite number of bilinear equations
	 	\begin{align}
		& \sum_{K,L,M=0}^{\infty}\big(d-b_2^{-1}+c\big) b_1^K b_2^L d^M S_K\left(\frac{1}{2} \tilde{D}_{+}\right) S_L\left(-\frac{1}{2} \tilde{D}_{+}\right) S_M\left(-\frac{1}{2} \tilde{D}_{-}\right) \tau_n \cdot \tau_{n+1} \notag\\
		& \qquad{}+\sum_{K,L,M=0}^{\infty}\big(b_1^{-1}-c-d\big) b_1^K b_2^L d^M S_K\left(-\frac{1}{2} \tilde{D}_{+}\right) S_L\left(\frac{1}{2} \tilde{D}_{+}\right) S_M\left(-\frac{1}{2}\tilde{D}_{-}\right) \tau_n \cdot \tau_{n+1} \notag\\
		& \qquad{}+\sum_{K,L,M=0}^{\infty}\big(b_2^{-1}-b_1^{-1}\big) b_1^K b_2^L d^M S_K\left(\frac{1}{2} \tilde{D}_{+}\right) S_L\left(\frac{1}{2} \tilde{D}_{+}\right) S_M\left(-\frac{1}{2} \tilde{D}_{-}\right) \tau_n \cdot \tau_{n+1}=0,\notag
		\end{align}
		where
		\begin{align*}
		\tilde{D}_{+}=\left(D_{x_1}, \frac{1}{2}D_{x_2},\dots,\frac{1}{n}D_{x_n},\dots\right),\qquad \tilde{D}_{-}=\left(D_{x_{-1}}, \frac{1}{2}D_{x_{-2}},\dots,\frac{1}{n}D_{x_{-n}},\dots\right).
		\end{align*}
		Here $S_n(\mathbf{x})$ are the elementary Schur polynomials which are defined via the generating function,{\samepage
		\begin{align*}
		\sum_{n=0}^{\infty}S_n(\mathbf{x})\lambda^n=\exp\left(\sum_{k=1}^{\infty}x_k\lambda^k\right),\qquad \mathbf{x}=(x_1,x_2,\dots).
		\end{align*}
		For example, we have
$S_0(\mathbf{x})=1$, $S_1(\mathbf{x})=x_1$, $S_2(\mathbf{x})=\frac{1}{2}x_1^2+x_2$, \dots.}

	At the order of $b_1^1b_2^0d^0$, we have
	 \begin{align*}
		&\left(-S_1\left(\frac{1}{2} \tilde{D}_{+}\right) S_1\left(-\frac{1}{2} \tilde{D}_{+}\right)+c S_1\left(\frac{1}{2} \tilde{D}_{+}\right)+S_2\left(-\frac{1}{2} \tilde{D}_{+}\right)\right.\\
		&\left.\qquad {}-c S_1\left(-\frac{1}{2} \tilde{D}_{+}\right)+S_1^2\left(\frac{1}{2} \tilde{D}_{+}\right)-S_2\left(\frac{1}{2} \tilde{D}_{+}\right)\right) \tau_n \cdot \tau_{n+1}=0,
		\end{align*}
	which gives equation \eqref{mkp1}.
	
	At the order of $b_1^1b_2^0 d^1$, we have
	 	\begin{align*}
		& \left(S_1\left(\frac{1}{2} \tilde{D}_{+}\right)-S_1\left(\frac{1}{2} \tilde{D}_{+}\right) S_1\left(-\frac{1}{2} \tilde{D}_{+}\right) S_1\left(-\frac{1}{2} \tilde{D}_{-}\right)+c S_1\left(\frac{1}{2} \tilde{D}_{+}\right) S_1\left(-\frac{1}{2} \tilde{D}_{-}\right)\right. \\
		& \qquad{}+S_2\left(-\frac{1}{2} \tilde{D}_{+}\right) S_1\left(-\frac{1}{2} \tilde{D}_{-}\right)-c S_1\left(-\frac{1}{2} \tilde{D}_{+}\right) S_1\left(-\frac{1}{2} \tilde{D}_{-}\right)-S_1\left(-\frac{1}{2} \tilde{D}_{+}\right) \\
		&\left.\qquad{}+S_1^2\left(\frac{1}{2} \tilde{D}_{+}\right) S_1\left(-\frac{1}{2} \tilde{D}_{-}\right)-S_2\left(\frac{1}{2} \tilde{D}_{+}\right) S_1\left(-\frac{1}{2} \tilde{D}_{-}\right)\right) \tau_n \cdot \tau_{n+1}=0,
		\end{align*}
	which leads to equation \eqref{mkp2}. The proof is complete.
\end{proof}

If we impose the constraints
$p_i=q_i$, $c_{ij}=\delta_{ij}$,
one can verify that $\partial_{x_2}\tau_n=0$. Setting $\tau_1=g$, $\tau_0=f$, we obtain from \eqref{mkp1}--\eqref{mkp2} the following bilinear equations
$\big(D_{x_1}^2+2 c D_{x_1}\big) f \cdot g=0$, $
 \big(D_{x_{-1}}\big(D_{x_1}^2+2 c D_{x_1}\big)-4 D_{x_1}\big) f \cdot g=0$.
Furthermore, by setting $x_1=y$, $x_{-1}=\frac{ {s}}{2}$ and $c=-\frac{1}{2}$, we arrive at the bilinear equations of the mCH equation \eqref{BL1}--\eqref{BL2}. Thus $\tau$-functions $f$ and $g$ admit the following Gram-type determinant form:
\begin{align*}
&\tau_n=\left|\delta_{ij}+\frac{1}{p_i+p_j}\left(-\frac{2p_i+1}{2p_j-1}\right)^{-n}{\rm e}^{\xi_i+\eta_j}\right|,\qquad
\xi_i=p_iy+\frac{1}{2p_i+1} {s}+\xi_{i0},\\
&\eta_j=p_jy+\frac{1}{2{{p_j}}-1} {s}+\eta_{j0},
\end{align*}
with $g=\tau_1$, $f=\tau_0$.

\section[Integrable semi-discretization of the modified Camassa--Holm equation]{Integrable semi-discretization\\ of the modified Camassa--Holm equation}\label{discretization}

In this section, we aim to construct the integrable spatial discretization of the mCH equation. To this end, we shall first derive semi-discrete analogs of the bilinear equations \eqref{BL1}--\eqref{BL2}. Subsequently, in Section \ref{dis-eqn}, we construct an integrable semi-discrete mCH equation.

\subsection[From discrete KP equation to the semi-discrete analog of (2.4) and (2.5)]{From discrete KP equation to the semi-discrete analog of (\ref{mkp1}) and (\ref{mkp2})}\label{dis-bl}
\begin{Lemma}
	The discrete KP equation \eqref{HM} generates the following bilinear equations:
	\begin{gather}
	 \bigl(-b_2^{-1}+c\bigr) \tau_n(k+1, l) \tau_{n+1}(k, l+1)+\big(b_1^{-1}-c\big) \tau_n(k, l+1) \tau_{n+1}(k+1, l)\nonumber \\
	\qquad{} +\big(b_2^{-1}-b_1^{-1}\big) \tau_n(k+1, l+1) \tau_{n+1}(k, l)=0,\label{d-dmkp1}
\\
	 \bigl(-b_2^{-1}+c\bigr) D_{ {x_{-1}}} \tau_n(k+1, l) \cdot \tau_{n+1}(k, l+1)+\big(b_1^{-1}-c\big) D_{ {x_{-1}}} \tau_n(k, l+1) \cdot \tau_{n+1}(k+1, l) \nonumber\\
	\qquad{} +\big(b_2^{-1}-b_1^{-1}\big) D_{ {x_{-1}}} \tau_n(k+1, l+1) \cdot \tau_{n+1}(k, l)+2(\tau_n(k, l+1) \tau_{n+1}(k+1, l)\nonumber\\
	\qquad{}-\tau_n(k+1, l) \tau_{n+1}(k, l+1))=0, \label{d-dmkp2}
	\end{gather}
	which admit the determinant solution of Gram-type
	\begin{align}
	\tau_n(k,l)&=\big|m_{ij}^{n,k,l}\big|\nonumber\\
&=\left|c_{ij}+\frac{1}{p_i+q_j}\left(-\frac{p_i-c}{q_j+c}\right)^{-n}\left(\frac{1-b_1 {p}_i}{1+b_1 {q}_j}\right)^{-k}\left(\frac{1-b_2 {p}_i}{1+b_2 {q}_j}\right)^{-l} {\rm e}^{\xi_i+\eta_j}\right|,\label{dis-sol-1}
	\end{align}
	where
$
	\xi_i=\frac{1}{p_i-c} x_{-1}+\xi_{i 0}$, $\eta_j=\frac{1}{q_j+c} x_{-1}+\eta_{j 0}$.
\end{Lemma}
\begin{proof}
	 	We apply the Miwa transformation to equation~\eqref{dmkp},
		\begin{align*}
		&x_{-1}={ {k_3}}d,\qquad { x_{-2}}=\frac{1}{2} { {k_3}}d^2,\qquad\dots, \qquad x_{-k}=\frac{1}{k} { {k_3}}d^k, \qquad {\ k=1,2,\dots,}
		\end{align*}
		then we have
		\begin{align*}
		& \sum_{M=0}^{\infty}\big(d-b_2^{-1}+c\big) d^M S_M\left(-\frac{1}{2} \tilde{D}_{-}\right) \tau_n(k_1+1, k_2) \cdot \tau_{n+1}(k_1, k_2+1) \\
		&\qquad {}+\sum_{M=0}^{\infty}\big(b_1^{-1}-c-d\big) d^M S_M\left(-\frac{1}{2} \tilde{D}_{-}\right) \tau_n(k_1, k_2+1) \cdot \tau_{n+1}(k_1+1, k_2) \\
		&\qquad {}+\sum_{M=0}^{\infty}\big(b_2^{-1}-b_1^{-1}\big) d^M S_M\left(-\frac{1}{2} \tilde{D}_{-}\right) \tau_n(k_1+1, k_2+1) \cdot \tau_{n+1}(k_1, k_2)=0.
		\end{align*}
	At the order of $d^0$ and $d^1$, we obtain equation \eqref{d-dmkp1} and \eqref{d-dmkp2} with $k_1=k$, $k_2=l$, respectively.
\end{proof}

\begin{Theorem}\label{dis-sol}
	Bilinear equations
	\begin{align}
	& \frac{1}{b}\left(f_{k+1} g_{k-1}-2 f_k g_k+f_{k-1} g_{k+1}\right)-\frac{1}{2}\left(f_{k+1} g_{k-1}-f_{k-1} g_{k+1}\right)=0 ,\label{dis-1}\\
	& \frac{2}{b} D_{ {s}}\left(f_{k+1} \cdot g_{k-1}-2 f_k \cdot g_k+f_{k-1} \cdot g_{k+1}\right) \nonumber\\
	& \qquad{}- D_{ {s}}\left(f_{k+1} \cdot g_{k-1}-f_{k-1} \cdot g_{k+1}\right)-2\left(f_{k+1} g_{k-1}-f_{k-1} g_{k+1}\right)=0.\label{dis-2}
	\end{align}
	admit the Gram-type determinant solution
	\begin{gather}
	f_k=\tau_0(k),\quad g_k=\tau_1(k),\nonumber\\
	\tau_n(k)=\big|{ m_{ij}^{n,k}}\big|=\left|\delta_{ij}+\frac{1}{p_i+p_j}\left(-\frac{2p_i+1}{2p_j-1}\right)^{-n}\left(\frac{1-b {p}_i}{1+b {p}_j}\right)^{-k}{\rm e}^{\xi_i+\eta_j}\right|,\label{gram2}
	\end{gather}
	where
	\begin{align}
	\xi_i=\frac{1}{2p_i+1} {s}+\xi_{i 0}, \qquad \eta_j=\frac{1}{2p_j-1} {s}+\eta_{j 0}.
	\end{align}
\end{Theorem}
\begin{proof}
	To realize the 2-reduction in the discrete case, we set
$b_1=-b_2=b$, $p_i=q_i$, $c_{ij}=\delta_{ij}$,
	in~\eqref{dis-sol-1}. Under these constraints, we have the reduction relation
$\tau_n(k+1,l+1)=\tau_n(k,l)$.
	From the reduction, we drop the index $l$ and define
$f_k=\tau_0(k)$, $g_k=\tau_1(k)$.
	Then from \eqref{d-dmkp1}--\eqref{d-dmkp2}, we have
	\begin{align}
	& \frac{1}{b}(f_{k+1} g_{k-1}-2 f_k g_k+f_{k-1} g_{k+1})+c(f_{k+1} g_{k-1}-f_{k-1} g_{k+1})=0 ,\label{s-b1}\\
	& \frac{1}{b} D_{ {x_{-1}}}(f_{k+1} \cdot g_{k-1}-2 f_k \cdot g_k+f_{k-1} \cdot g_{k+1}) \nonumber\\
	& \qquad{}+c D_{ {x_{-1}}}(f_{k+1} \cdot g_{k-1}-f_{k-1} \cdot g_{k+1})-2(f_{k+1} g_{k-1}-f_{k-1} g_{k+1})=0.\label{s-b2}
	\end{align}
	By setting $x_{-1}=\frac{ {s}}{2}$ and $c=-\frac{1}{2}$, equations \eqref{s-b1}--\eqref{s-b2} are transformed into \eqref{dis-1}--\eqref{dis-2}. Gram determinant solution \eqref{gram2} can be obtained directly by using the reduction from \eqref{dis-sol-1}.
\end{proof}

\subsection{Integrable semi-discretization of the mCH equation}\label{dis-eqn}
Based on the semi-discrete bilinear equations in Theorem \ref{dis-sol}, we propose an integrable semi-discrete mCH equation.
\begin{Theorem}\label{dmch}
	An integrable semi-discrete analogue of the mCH equation \eqref{mch} is derived as
	\begin{align}
	&\frac{\mathrm{d} {m}^{-1}_k}{\mathrm{d} t}=2 {m}_k \Gamma_k { \left(\frac{ u_{k+1}-u_k}{b}\right)},\label{dmch-1}\\
	&m_k =\frac{u_{k+1}+u_k}{2}-\frac{1}{2} {m}_k\left(1+\frac{b^2}{4}\big(m_k^{-1}-1\big)\right){ \frac{\mathrm{d}}{\mathrm{d} t}\left(\frac{\tilde{m}_k^{-1}
			-\tilde{m}_{k-1}^{-1}}{b}\right)},\label{dmch-2}
	\end{align}
	from equations \eqref{dis-1}--\eqref{dis-2} through a dependent variable transformation
\smash{$
	u_k=1-\frac{1}{b}\bigl(\ln \frac{g_k f_k}{g_{k-1} f_{k-1}}\bigr)_{ {s}}$},
	and a discrete reciprocal transformation
	\begin{align*}
	\delta x_k=\frac{x_{k+1}-x_k}{b}=1+\frac{2}{b} \frac{f_{k-1} g_{k+1}-f_{k+1} g_{k-1}}{f_{k-1} g_{k+1}+f_{k+1} g_{k-1}},\qquad x_k=x_0+b\sum_{l=0}^{k-1} \delta x_l,\qquad t={ {s}}.
	\end{align*}
	Other variables are defined by
	\begin{align*}
	&\tilde{x}_k=k b+{ {s}}+2 \ln \frac{g_k}{f_k},\qquad
	{m}_k^{-1}=\delta x_k=1+\frac{2}{b} \frac{f_{k-1} g_{k+1}-f_{k+1} g_{k-1}}{f_{k-1} g_{k+1}+f_{k+1} g_{k-1}},\nonumber\\
	&\delta u_k=\frac{u_{k+1}-u_k}{b}=-\frac{1}{b^2}\left(\ln \frac{f_{k+1} g_{k-1} f_{k-1} g_{k+1}}{f_k^2 g_k^2}\right)_{ {s}},\\
	&\tilde{m}_k^{-1}=\delta \tilde{x}_k=\frac{\tilde{x}_{k+1}-\tilde{x}_k}{b}=1+\frac{2}{b} \ln \frac{g_{k+1} f_k}{g_k f_{k+1}}, \\
	&\Gamma_k=1+\frac{m_k^{-1}-1}{4}b-\frac{\big(m_k^{-1}-1\big)^2}{4}b^2-\frac{\big(m_k^{-1}-1\big)^3}{16}b^3,\\
	&\delta\big(\tilde{m}_k^{-1}\big)=\frac{\tilde{m}_k^{-1}-\tilde{m}_{k-1}^{-1}}{b}=-\frac{2}{b^2} \ln \frac{f_{k+1} f_{k-1} g_k^2}{g_{k+1} g_{k-1} f_k^2}.
	\end{align*}
\end{Theorem}
Prior to the proof of the theorem, we show that the semi-discrete mCH equations \eqref{dmch-1}--\eqref{dmch-2} converge to the mCH equation \eqref{mch} in the continuous limit $b\rightarrow 0$.

Recall that
$u=1-(\ln f g)_{y { {s}}}$, $x=y+{ {s}}+2 \ln \frac{g}{f}$.
It is obvious that when $b\rightarrow 0$ we have
\begin{align*}
u_k\rightarrow u,\qquad \tilde{x}_k\rightarrow x,\qquad \delta u_k\rightarrow u_y,\qquad \tilde{m}_k^{-1}\rightarrow\frac{\partial x}{\partial y}=\frac{1}{m},\qquad \delta\big(\tilde{m}_k^{-1}\big)\rightarrow \left(\frac{1}{m}\right)_y,
\end{align*}
and furthermore,
$\Gamma_k\rightarrow 1$, $f_{k+1}\rightarrow f_k+bf_{k,y}$, $g_{k+1}\rightarrow g_k+bg_{k,y}$,
which leads to
\begin{align*}
\frac{2}{b} \frac{f_{k-1} g_{k+1}-f_{k+1} g_{k-1}}{f_{k-1} g_{k+1}+f_{k+1} g_{k-1}}\rightarrow \frac{2}{b} \frac{f_{k-1}(g_{k-1}+2bg_{k,y})-g_{k-1}(f_{k-1}+2bf_{k,y})}{f_{k-1}(g_{k-1}+2bg_{k,y})+g_{k-1}(f_{k-1}+2bf_{k,y})}\rightarrow 2\left(\ln\frac{g}{f}\right)_y.
\end{align*}
Therefore, we have
\begin{align*}
\delta x_k=m_k^{-1}\rightarrow 1+2\left(\ln\frac{g}{f}\right)_y=\frac{\partial x}{\partial y}=\frac{1}{m}.
\end{align*}
Thus, we conclude that equations \eqref{dmch-1}--\eqref{dmch-2} converge to
\begin{align}
&\frac{\partial^2x}{\partial y\partial { {s}}}=\left(\frac{1}{m}\right)_{ {s}}=2mu_y,\qquad
m=u-\frac{1}{2}m\left(\frac{1}{m}\right)_{y{ {s}}}=u-m(mu_y)_y=u-u_{xx},\label{mch-x}
\end{align}
respectively. On the other hand, equation \eqref{mch-x} is equivalent to
\begin{align*}
\frac{\partial^2x}{\partial y\partial { {s}}}=2mu_y=2(u-m(mu_y)_y)u_y=\big(u^2-m^2u_y^2\big)_y=\big(u^2-u_x^2\big)_y,
\end{align*}
or
$\frac{\partial x}{\partial { {s}}}=u^2-u_x^2$,
which implies
$\partial_{ {s}}=\partial_t+\big(u^2-u_x^2\big)\partial_x$.
As a result, equation~\eqref{mch-x} leads to
\begin{align*}
m_{ {s}}+2m^3u_y=&m_t+\big(u^2-u_x^2\big)m_x+2m^2u_x
=m_t+\big[m\big(u^2-u_x^2\big)\big]_x=0,
\end{align*}
which is actually the mCH equation \eqref{mch}.

In the following, we present the detailed proof of the theorem.
\begin{proof}
	We rewrite equation~\eqref{dis-1} as
\begin{gather*}
	\frac{1}{b}\left(\frac{f_{k+1} g_{k-1}}{f_k g_k}-2+\frac{f_{k-1} g_{k+1}}{f_k g_k}\right)-\frac{1}{2}\left(\frac{f_{k+1} g_{k-1}}{f_k g_k}+\frac{f_{k-1} g_{k+1}}{f_k g_k}\right) \frac{f_{k+1} g_{k-1}-f_{k-1} g_{k+1}}{f_{k+1} g_{k-1}+f_{k-1} g_{k+1}}=0,
	\end{gather*}
	or equivalently
	\begin{align}
	-\frac{2}{b}+\frac{f_{k+1} g_{k-1}+f_{k-1} g_{k+1}}{f_k g_k}\left(\frac{1}{b}-\frac{1}{2} \frac{f_{k+1} g_{k-1}-f_{k-1} g_{k+1}}{f_{k+1} g_{k-1}+f_{k-1} g_{k+1}}\right)=0.\label{dis-3}
	\end{align}
	By using the identity $\rho_{ {s}}=\rho\left(\ln \rho\right)_{ {s}}$, we have
	\begin{gather*}
	 \left(\frac{f_{k+1} g_{k-1}+f_{k-1} g_{k+1}}{f_k g_k}\right)_{ {s}} \\
	\qquad =\frac{f_{k+1} g_{k-1}}{f_k g_k}\left(\ln \frac{f_{k+1} g_{k-1}}{f_k g_k}\right)_{ {s}}+\frac{f_{k-1} g_{k+1}}{f_k g_k}\left(\ln \frac{f_{k-1} g_{k+1}}{f_k g_k}\right)_{ {s}} \\
	\qquad =\frac{f_{k+1} g_{k-1}+f_{k-1} g_{k+1}}{2 f_k g_k}\left(\ln \frac{f_{k+1} g_{k-1} f_{k-1} g_{k+1}}{f_k^2 g_k^2}\right)_{ {s}}\\
	\phantom{\qquad=}{}+\frac{f_{k-1} g_{k+1}-f_{k+1} g_{k-1}}{2 f_k g_k}\left(\ln \frac{f_{k-1} g_{k+1}}{f_{k+1} g_{k-1}}\right)_{ {s}},
	\end{gather*}
	and
	\begin{align*}
	& \left(\frac{f_{k+1} g_{k-1}-f_{k-1} g_{k+1}}{f_{k+1} g_{k-1}+f_{k-1} g_{k+1}}\right)_{ {s}}
	=-\frac{2 f_{k+1} g_{k-1} f_{k-1} g_{k+1}}{\left(f_{k+1} g_{k-1}+f_{k-1} g_{k+1}\right)^2}\left(\ln \frac{f_{k-1} g_{k+1}}{f_{k+1} g_{k-1}}\right)_{ {s}}.
	\end{align*}
	Therefore, differentiating equation~\eqref{dis-3} with respect to ${ {s}}$ leads to
	\begin{gather*}
	 \left[\frac{f_{k+1} g_{k-1}+f_{k-1} g_{k+1}}{2 f_k g_k}\left(\ln \frac{f_{k+1} g_{k-1} f_{k-1} g_{k+1}}{f_k^2 g_k^2}\right)_{ {s}}\right.\\
\left.	\qquad{}+\frac{f_{k-1} g_{k+1}-f_{k+1} g_{k-1}}{2 f_k g_k}\left(\ln \frac{f_{k-1} g_{k+1}}{f_{k+1} g_{k-1}}\right)_{ {s}}\right]\cdot \left(\frac{1}{b}-\frac{1}{2} \frac{f_{k+1} g_{k-1}-f_{k-1} g_{k+1}}{f_{k+1} g_{k-1}+f_{k-1} g_{k+1}}\right)\\
	\qquad {}+\frac{f_{k+1} g_{k-1}+f_{k-1} g_{k+1}}{f_k g_k} \frac{f_{k+1} g_{k-1} f_{k-1} g_{k+1}}{\left(f_{k+1} g_{k-1}+f_{k-1} g_{k+1}\right)^2}\left(\ln \frac{f_{k-1} g_{k+1}}{f_{k+1} g_{k-1}}\right)_{ {s}}=0.
	\end{gather*}
	
	Dividing both sides by $\frac{f_{k+1} g_{k-1}+f_{k-1} g_{k+1}}{2 f_k g_k}$, we have
	\begin{gather}
	 \left(\frac{1}{b}-\frac{1}{2} \frac{f_{k+1} g_{k-1}-f_{k-1} g_{k+1}}{f_{k+1} g_{k-1}+f_{k-1} g_{k+1}}\right)\left(\ln \frac{f_{k+1} g_{k-1} f_{k-1} g_{k+1}}{f_k^2 g_k^2}\right)_{ {s}}\nonumber\\
	\qquad{}+\frac{1}{b} \frac{f_{k-1} g_{k+1}-f_{k+1} g_{k-1}}{f_{k-1} g_{k+1}+f_{k+1} g_{k-1}}\left(\ln \frac{f_{k-1} g_{k+1}}{f_{k+1} g_{k-1}}\right)_{ {s}} +\frac{1}{2} \left(\ln \frac{f_{k-1} g_{k+1}}{f_{k+1} g_{k-1}}\right)_{ {s}}=0.\label{dis-4}
	\end{gather}
	As $b \rightarrow 0$, equation~\eqref{dis-4} converges to
	\[
	(\ln f g)_{y y { {s}}}+2\left(\ln \frac{g}{f}\right)_y\left(\ln \frac{g}{f}\right)_{y { {s}}}+\left(\ln \frac{g}{f}\right)_{y { {s}}}=0 .
	\]
	
	From the definition of $\Gamma_k$, $u_k$, $\delta u_k$, and $m_k^{-1}$, we have
	\begin{align*}
	\Gamma_k &=1+\frac{m_k^{-1}-1}{4}b-\frac{\big(m_k^{-1}-1\big)^2}{4}b^2-\frac{\big(m_k^{-1}-1\big)^3}{16}b^3\\
	&=\left(1-\frac{b}{2} \frac{f_{k+1} g_{k-1}-f_{k-1} g_{k+1}}{f_{k+1} g_{k-1}+f_{k-1} g_{k+1}}\right) \left(1-\left(\frac{f_{k+1} g_{k-1}-f_{k-1} g_{k+1}}{f_{k+1} g_{k-1}+f_{k-1} g_{k+1}}\right)^2\right)\\
	& =\left(1-\frac{b}{2} \frac{f_{k+1} g_{k-1}-f_{k-1} g_{k+1}}{f_{k+1} g_{k-1}+f_{k-1} g_{k+1}}\right) \frac{4 f_{k+1} g_{k-1} f_{k-1} g_{k+1}}{\left(f_{k+1} g_{k-1}+f_{k-1} g_{k+1}\right)^2}.
	\end{align*}
	Then equation~\eqref{dis-4} leads to
	\begin{align}
	\big({m}_k^{-1}\big)_{ {s}}=2 {m}_k \Gamma_k \delta u_k.\label{dis-4'}
	\end{align}
	Since $\delta x_k=m_k^{-1}$, one can rewrite equation~\eqref{dis-4'} as
$
	\frac{\mathrm{d} \delta x_k}{\mathrm{d} t}=2 {m}_k \Gamma_k (\delta u_k)$,
	which constitutes the first equation of the semi-discrete mCH equation. Now we are ready to deduce the second equation of the semi-discrete mCH equation.
	We rewrite equation~\eqref{dis-2} into
	\begin{align*}
	& \frac{1}{b}\left(f_{k+1} g_{k-1}\left(\ln \frac{f_{k+1}}{g_{k-1}}\right)_{ {s}}-2 f_k g_k\left(\ln \frac{f_k}{g_k}\right)_{ {s}}+f_{k-1} g_{k+1}\left(\ln \frac{f_{k-1}}{g_{k+1}}\right)_{ {s}}\right) \\
	& \qquad{}-\frac{1}{2}\left(f_{k+1} g_{k-1}\left(\ln \frac{f_{k+1}}{g_{k-1}}\right)_{ {s}}-f_{k-1} g_{k+1}\left(\ln \frac{f_{k-1}}{g_{k+1}}\right)_{ {s}}\right)-\left(f_{k+1} g_{k-1}-f_{k-1} g_{k+1}\right)=0.
	\end{align*}
	Thus, we have
	\begin{align}
	& \frac{1}{b}\left(\ln \frac{f_{k+1} f_{k-1}}{g_{k+1} g_{k-1}}\right)_{ {s}}-\frac{4}{b} \frac{f_k g_k}{f_{k+1} g_{k-1}+f_{k-1} g_{k+1}}\left(\ln \frac{f_k}{g_k}\right)_{ {s}}\nonumber\\
	&\qquad{}+\frac{1}{b} \frac{f_{k+1} g_{k-1}-f_{k-1} g_{k+1}}{f_{k+1} g_{k-1}+f_{k-1} g_{k+1}}\left(\ln \frac{f_{k+1} g_{k+1}}{f_{k-1} g_{k-1}}\right)_{ {s}} -\frac{1}{2}\left(\ln \frac{f_{k+1} g_{k+1}}{f_{k-1} g_{k-1}}\right)_{ {s}}\nonumber \\
	&\qquad{}-\frac{1}{2}\frac{f_{k+1} g_{k-1}-f_{k-1} g_{k+1}}{f_{k+1} g_{k-1}+f_{k-1} g_{k+1}}\left(\ln \frac{f_{k+1} f_{k-1}}{g_{k+1} g_{k-1}}\right)_{ {s}}-2 \frac{f_{k+1} g_{k-1}-f_{k-1} g_{k+1}}{f_{k+1} g_{k-1}+f_{k-1} g_{k+1}}=0.\label{dis-5}
	\end{align}
	By rewriting equation~\eqref{dis-1} as
\[
	\frac{1}{b} \frac{2 f_k g_k}{f_{k+1} g_{k-1}+f_{k-1} g_{k+1}}=\frac{1}{b}-\frac{1}{2}\frac{f_{k+1} g_{k-1}-f_{k-1} g_{k+1}}{f_{k+1} g_{k-1}+f_{k-1} g_{k+1}},
	\]
	and substituting it into equation~\eqref{dis-5}, one obtains
	\begin{align}
	& \frac{1}{b}\frac{2 f_k g_k}{f_{k+1} g_{k-1}+f_{k-1} g_{k+1}}\left(\ln \frac{f_{k+1} f_{k-1} g_k^2}{g_{k+1} g_{k-1} f_k^2}\right)_{ {s}}+\frac{1}{b} \frac{f_{k+1} g_{k-1}-f_{k-1} g_{k+1}}{f_{k+1} g_{k-1}+f_{k-1} g_{k+1}}\left(\ln \frac{f_{k+1} g_{k+1}}{f_{k-1} g_{k-1}}\right)_{ {s}} \notag\\
	& \qquad{}-\frac{1}{2}\left(\ln \frac{f_{k+1} g_{k+1}}{f_{k-1} g_{k-1}}\right)_{ {s}}-2 \frac{f_{k+1} g_{k-1}-f_{k-1} g_{k+1}}{f_{k+1} g_{k-1}+f_{k-1} g_{k+1}}=0.\label{dis-6}
	\end{align}
	From the definition of $u_k$ and $\delta\big(\tilde{m}_k^{-1}\big)$, one can obtain
\smash{$u_{k+1}+u_k=2-\frac{1}{b}\bigl(\ln \frac{g_{k+1} f_{k+1}}{g_{k-1} f_{k-1}}\bigr)_{ {s}}$}
	equation~\eqref{dis-6} can be rewritten as
	\begin{align*}
	-\frac{b}{2}\left(1+\frac{b^2}{4}\big(m_k^{-1}-1\big)\right)\left(\delta\big(\tilde{m}_k^{-1}\big)\right)_{ {s}}-bm_k^{-1}\left(1-\frac{u_k+u_{k+1}}{2}\right)+b\big(m_k^{-1}-1\big)=0,
	\end{align*}
	which can be shown equivalent to equation~\eqref{dmch-2}. The proof is completed.
\end{proof}

 From Theorems \ref{dis-sol} and \ref{dmch}, we find that the semi-discrete mCH equation \eqref{dmch-1}--\eqref{dmch-2} admits $N$-soliton solution in the determinant form
\begin{align*}
& u_k=1-\frac{1}{b}\left(\ln \frac{g_k f_k}{g_{k-1} f_{k-1}}\right)_{ {s}},\qquad \delta x_k\equiv\frac{x_{k+1}-x_k}{b}=1+\frac{2}{b} \frac{f_{k-1} g_{k+1}-f_{k+1} g_{k-1}}{f_{k-1} g_{k+1}+f_{k+1} g_{k-1}},\\
& m_k=(\delta x_k)^{-1},
\end{align*}
where $f_k$, $g_k$ are given by the determinant \eqref{gram2}.

\subsection{One- and two-soliton solutions}
\subsubsection{One-soliton solutions}
The $\tau$-functions for the one-soliton solution of the semi-discrete mCH equation in Theorem \ref{dmch} are\looseness=-1
\begin{align*}
f_k\propto1+\left(\frac{1-bp}{1+bp}\right)^{-k}{\rm e}^{\zeta},\qquad g_k\propto1+\left(-\frac{2p+1}{2p-1}\right)^{-1}\left(\frac{1-bp}{1+bp}\right)^{-k}{\rm e}^{\zeta},
\end{align*}
with $\zeta=-\frac{4p}{1-4p^2}{ {s}}+\zeta_0$. Here we set $p = p_1$ for simplicity. Thus, we can obtain the one-soliton solution in a parametric form
\begin{align*}
u_k&=1-\frac{1}{b}\left(\ln \frac{g_k f_k}{g_{k-1} f_{k-1}}\right)_{ {s}}=1-\frac{1}{b}\frac{4p}{1-4p^2}\left(\frac{1}{f_k}+\frac{1}{g_k}-\frac{1}{f_{k-1}}-\frac{1}{g_{k-1}}\right),\\
x_k&=x_0+b\sum_{i=0}^{k-1}\delta x_i=x_0+(k-1)b+2 \sum_{i=0}^{k-1}\frac{f_{i-1} g_{i+1}-f_{i+1} g_{i-1}}{f_{i-1} g_{i+1}+f_{i+1} g_{i-1}}.
\end{align*}
When we take $b=0.1$, $\zeta_0=0$, and choose appropriate $x_0$ such that the solution $u_k$ is symmetric with respect to $x_k$, Figure \ref{dis-1-soliton-fig} displays two different kinds of solutions for the semi-discrete mCH equation under different $p$ values. Figure \ref{1-soliton-fig} depicts a one-soliton solution to the semi-discrete mCH equation while comparing with the one-soliton solution to the mCH equation. When~\smash{$0<|p|<\frac{\sqrt{3}}{4}$}, the solution $u_k$ is single-valued with one peak since $\delta_k>0$ (see Figure~\ref{dis-mch-1}). Figure~\ref{dis-mch-2} illustrates the symmetric singular soliton solutions that are three-valued with two spikes for \smash{$\frac{\sqrt{3}}{4}<|p|<\frac{1}{2}$}. Figure~\ref{1-soliton-fig} shows the comparison among the one-soliton solutions for the mCH equation in \cite{hu-soliton,matsuno1} and the semi-discrete mCH equation at $t=0$. It should be pointed out that the semi-discrete analogue of the mCH equation with linear dispersion term admits anti-symmetric singular soliton solutions (see \cite[Figures 1C and 2C]{mch-SYF}), while the semi-discrete mCH equation without linear dispersion term we proposed here does not admit such singular solution.

\begin{figure}[t]
	\centering
	\subfigure[Smooth soliton solutions]
	{
		\label{dis-mch-1}
		\includegraphics[width=1.5in]{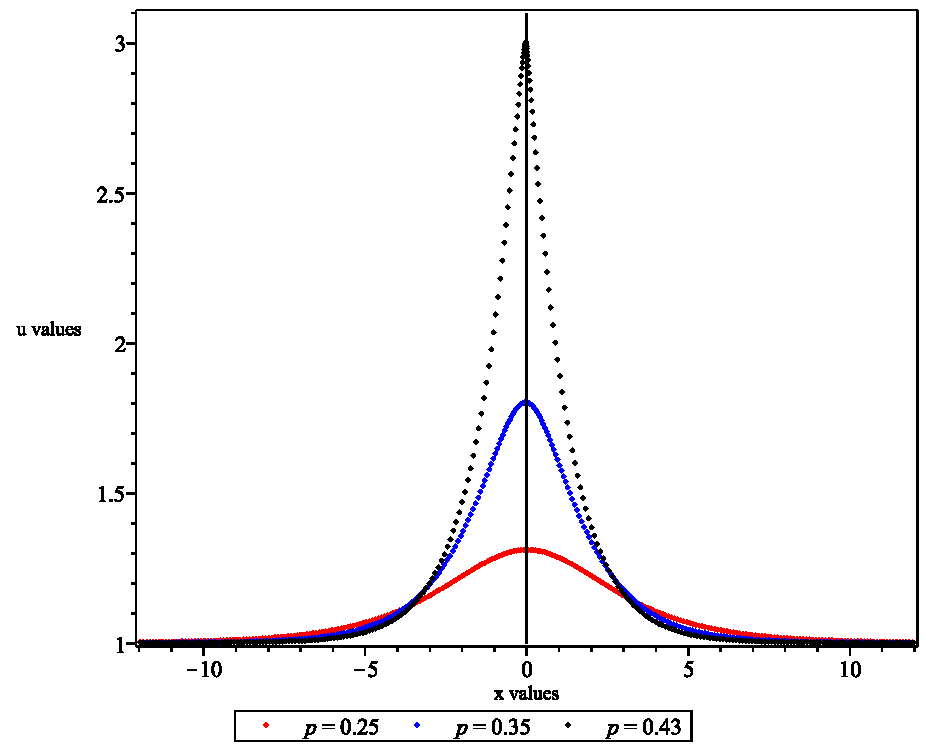}
	}
	\hspace*{3em}
	\subfigure[Symmetric singular soliton solutions]
	{\label{dis-mch-2}
		\includegraphics[width=1.5in]{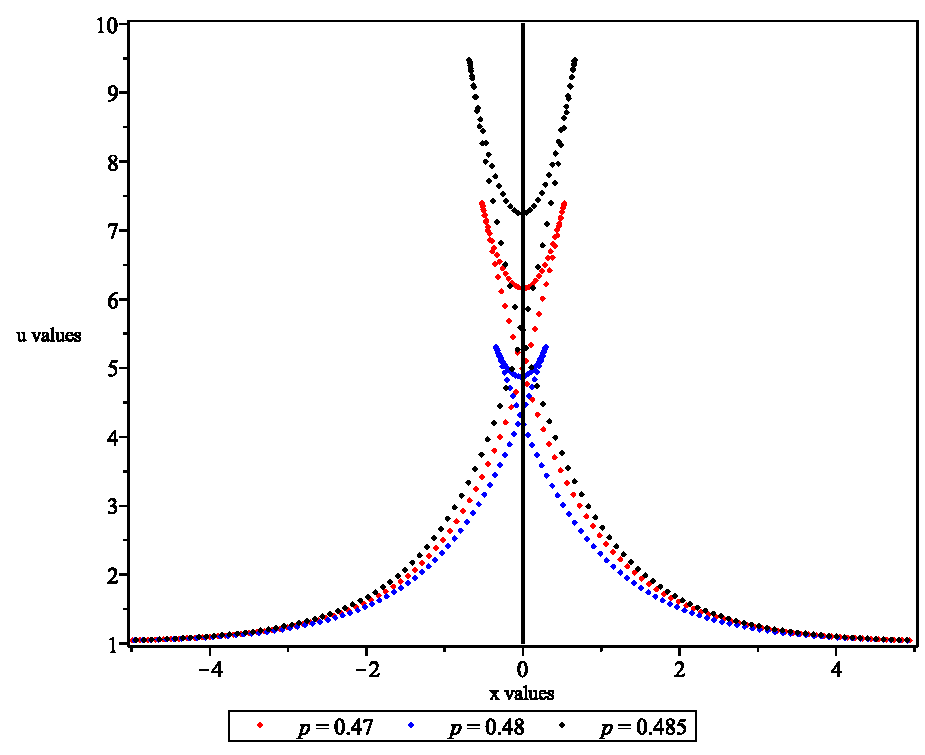}
	}
	\caption{Two different kinds of solutions for the semi-discrete mCH equation at $t=0$. (a) Smooth soliton solutions, (b) Symmetric singular soliton solutions.}\label{dis-1-soliton-fig}
\end{figure}
\begin{figure}
	\centering
	\subfigure[Smooth soliton solutions]
	{
		\label{mch-1}
		\includegraphics[width=1.5in]{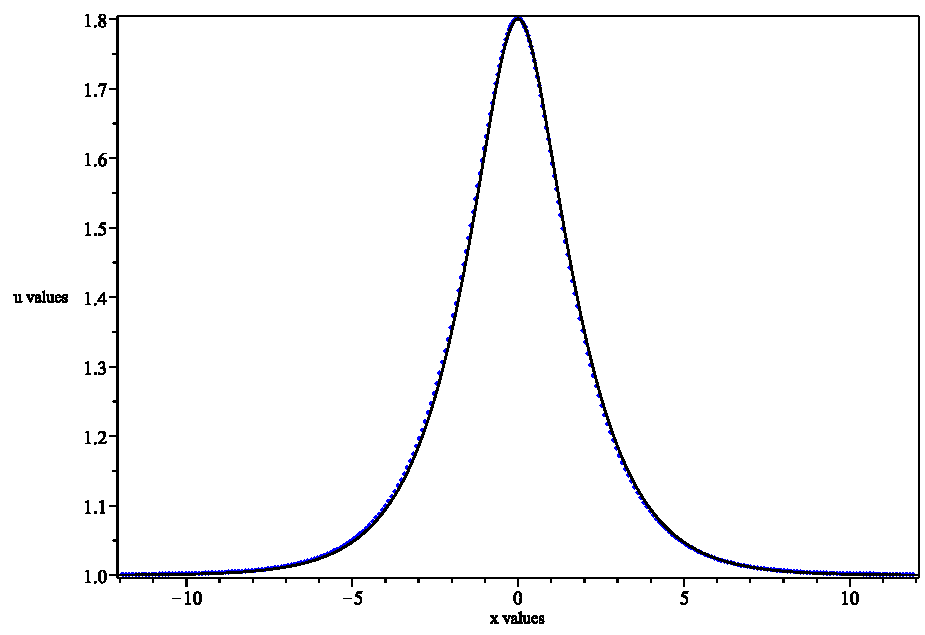}
	}
	\hspace*{3em}
	\subfigure[Symmetric singular soliton solutions]
	{\label{mch-2}
		\includegraphics[width=1.5in]{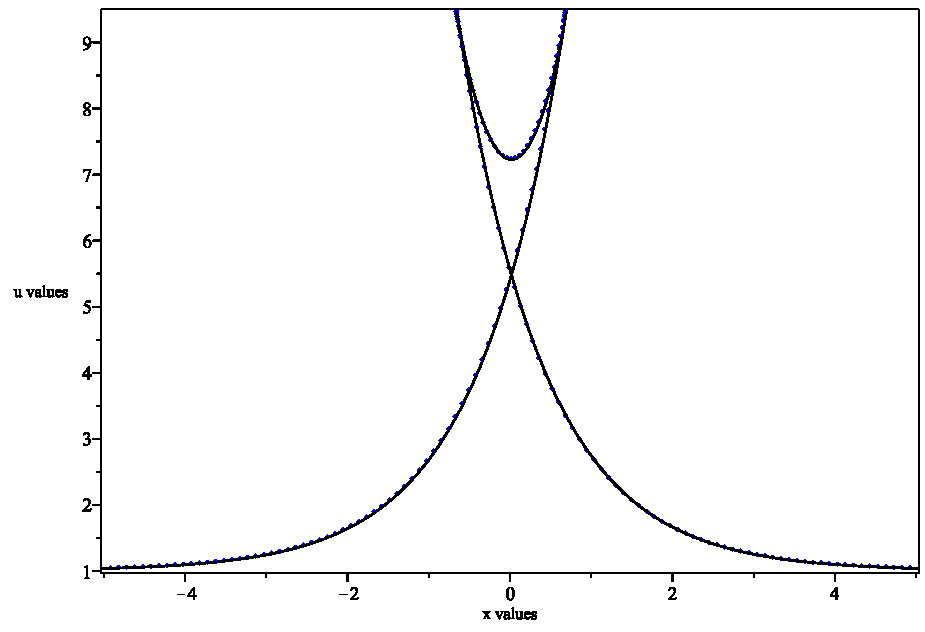}
	}
	\caption{Comparison between the one-soliton solution for the mCH equation and the semi-discrete mCH equation at $t=0$; solid line: mCH equation, dot: semi-discrete mCH equation. (a)~$p=0.35$, (b)~$p=0.485$.}\label{1-soliton-fig}
\end{figure}

\subsubsection{Two-soliton solutions}
The $\tau$-functions for the two-soliton solution of the semi-discrete mCH equation in Theorem \ref{dmch} are
\begin{align*}
f_k\propto& 1+z_1^{-k}{\rm e}^{\zeta_1}+z_2^{-k}{\rm e}^{\zeta_2}+\left(\frac{p_1-p_2}{p_1+p_2}\right)^2(z_1z_2)^{-k}{\rm e}^{\zeta_1+\zeta_2},\\
g_k\propto& 1+\frac{1-2p_1}{1+2p_1}z_1^{-k}{\rm e}^{\zeta_1}+\frac{1-2p_2}{1+2p_2}z_2^{-k}{\rm e}^{\zeta_2}+\frac{1-2p_1}{1+2p_1}\frac{1-2p_2}{1+2p_2}\left(\frac{p_1-p_2}{p_1+p_2}\right)^2(z_1z_2)^{-k}{\rm e}^{\zeta_1+\zeta_2},
\end{align*}
with \smash{$z_i=\frac{1-bp_i}{1+bp_i}$} and \smash{$\zeta_i=-\frac{4p_i}{1-4p_i^2}{ {s}}+\zeta_{i0}$}. We take $b=0.1$ and $\zeta_{i0}=0$. Figure \ref{2-soliton-fig} displays the collision between two smooth solitons. One can see that the soliton with a higher peak moves faster than the lower one. It can be found that there is a strong agreement between the two-soliton solution of the semi-discrete mCH equation and the mCH equation.
\begin{figure}
	\centering
	\subfigure[$t=-15$]
	{
		\label{2sol1}
		\includegraphics[width=1.5in]{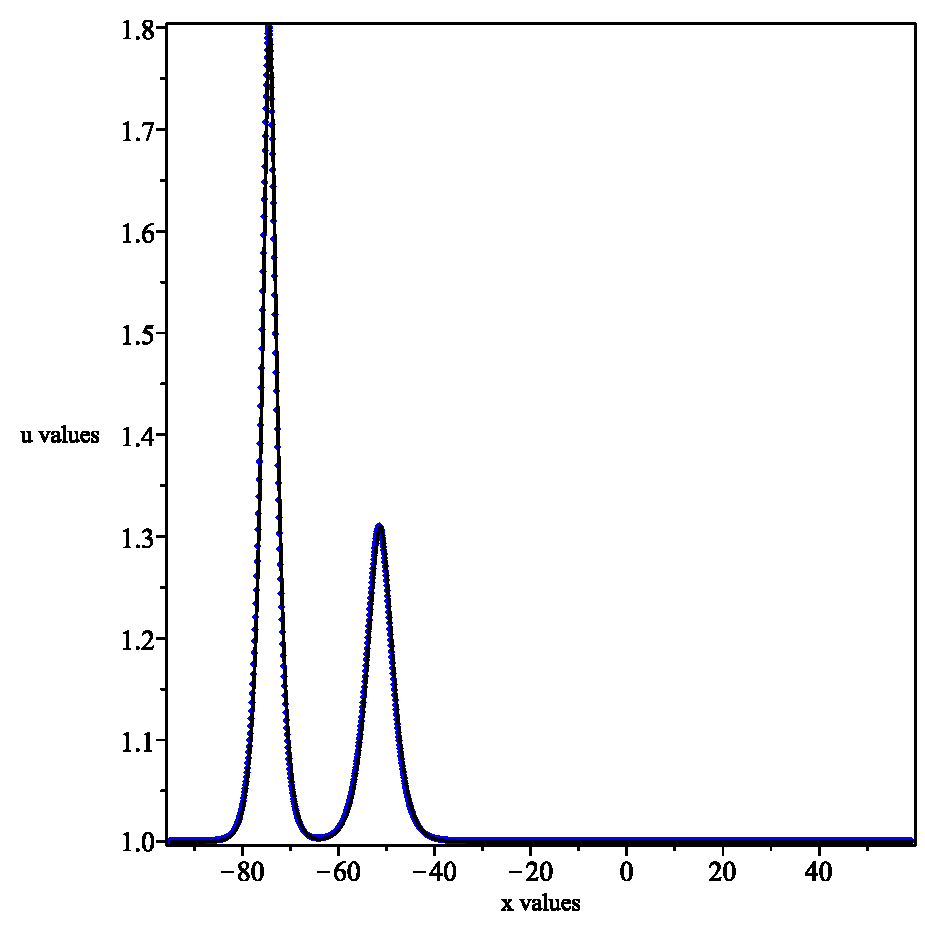}
	}
	\hspace*{3em}
	\subfigure[$t=0$]
	{\label{2sol2}
		\includegraphics[width=1.5in]{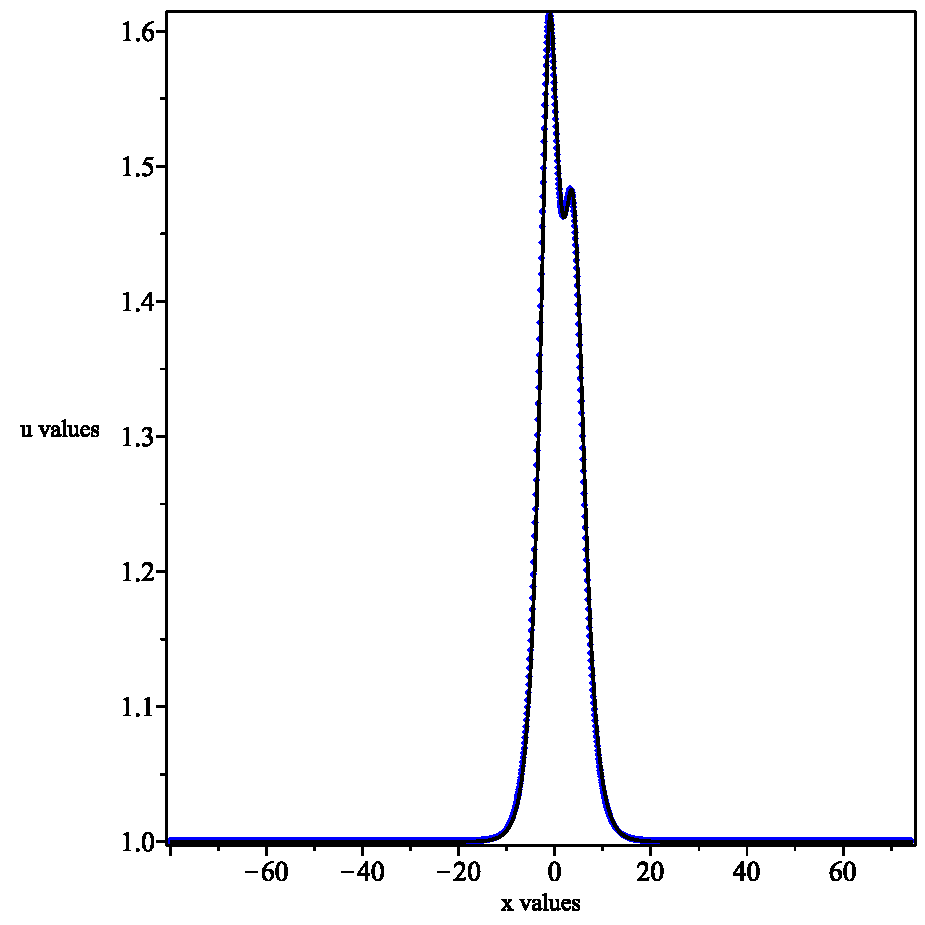}
	}
	\hspace*{3em}
	\subfigure[$t=15$]
	{\label{2sol3}
		\includegraphics[width=1.5in]{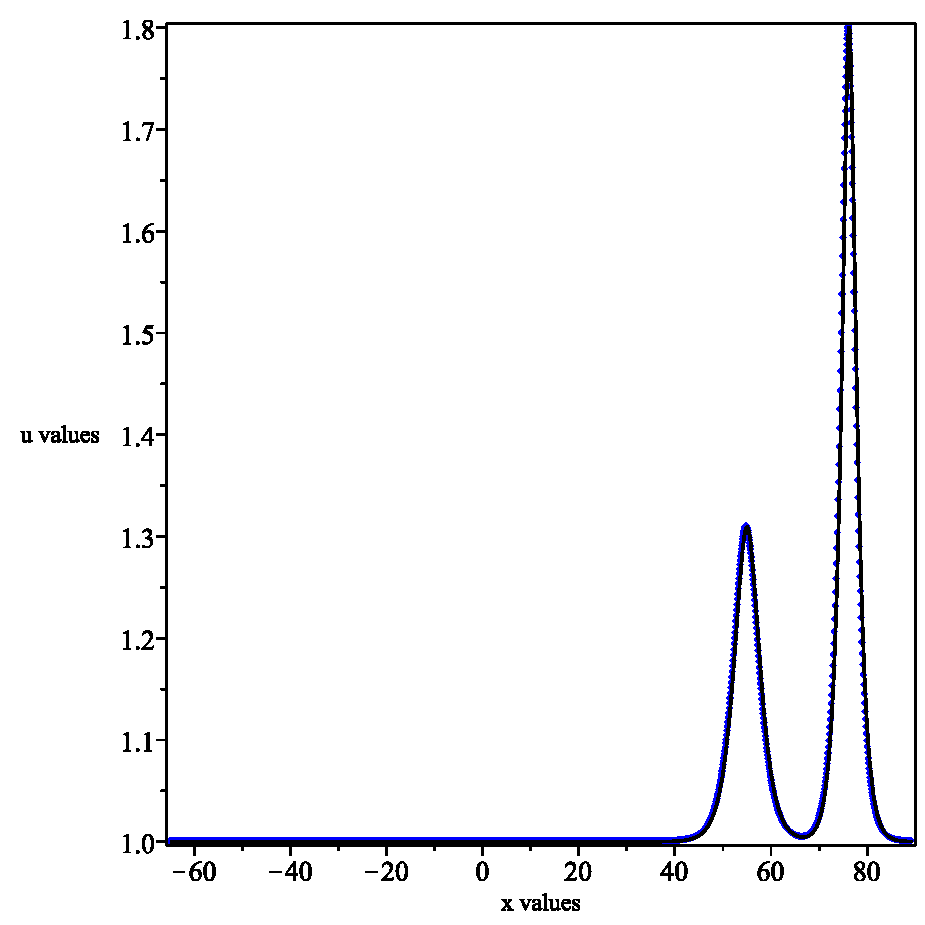}
	}
	\caption{Comparison between the two-soliton solution of the mCH and the semi-discrete mCH
		equation with $p_1=0.25$, $p_2=0.35$; solid line: mCH equation; dot: semi-discrete mCH equation. (a)~$t=-15$, (b)~$t=0$, (c)~$t=15$.}\label{2-soliton-fig}
\end{figure}

It is shown that the proposed semi-discrete mCH equation, and one and two-soliton solution converge to ones of the original mCH equation in the continuum limit. This is the reason that there is no substantial difference between soliton solutions of the mCH equation and its semi-discrete version. The situation also holds for the semi-discrete integrable DP equation, generalized sine-Gordon equation and short pulse equation.

\section{Conclusion and discussion}\label{conclusion}

In this paper, starting from the discrete KP equation, we have constructed an integrable semi-discrete analog of the mCH equation with cubic nonlinearity through Miwa transformation and a series of reductions. Gram-type determinant solutions for the semi-discrete mCH equation has been derived. Smooth soliton solutions and symmetric singular soliton solutions are generated from the determinant formulas. The discrete KP equation is once again shown to be the fundamental equation for integrable systems, in line with the findings by Hirota, Ohta, Tsujimoto, Nimmo, and so on. Furthermore, there are a few aspects that deserve further study. Firstly, the Lax pair associated with the semi-discrete mCH equation is still unknown. How to generate the Lax pair for the derived discrete integrable systems based on the Lax pair of discrete KP equation is left to be investigated. Secondly, here we only find semi-discrete version of the mCH equation and the full-discrete analogue of the mCH is left to be considered. Thirdly, connections between the discrete KP equation and the two-component CH equation~\cite{YoujinLMP}, the two-component mCH equation~\cite{SQQ}, the complex short pulse equation~\cite{Feng15} and the massive Thirring model equation~\cite{MikhailovMT,Thirring} are worth investigating.

 DP equation and Novikov equation are peakon-type integrable nonlinear partial differential equations with higher-order nonlinearity. An integrable semi-discrete DP equation has been constructed from a pseudo 3-reduction of the CKP hierarchy \cite{dDP}. Integrable semi-discretizations for the short wave limit of the Novikov equation was presented in \cite{Feng_swNV}. To the best of our knowledge, integrable discrete analogues of the Novikov equation based on the methodology here have not been reported. Furthermore, how to apply the proposed integrable semi-discrete mCH equation as a self-adaptive moving mesh scheme for the numerical simulation of the mCH equation deserves further exploration. These intriguing topics will be addressed in our future study.\looseness=-1

An infinite number of conservation laws of the mCH equation was found in \cite{matsuno1} based on the B{\"a}cklund transformation, which read as
\[
I_{2n+1}=\sum_{m=0}^{n+1}\nu_{nm}\tilde{I}_m.
\] Here $\nu_{nm}$ are constants depending on the constant background $u_0$, and the first three of $\tilde{I}_m$ are expressed as
	\begin{gather*}
	\tilde{I}_0=\int_{-\infty}^{\infty}\left(m-u_0\right) \mathrm{d} x, \qquad
	\tilde{I}_1=\int_{-\infty}^{\infty}\left(\frac{1}{m}-\frac{1}{u_0}\right) \mathrm{d} x,\\
	\tilde{I}_2=\int_{-\infty}^{\infty}\left[\frac{1}{m^3}-\frac{1}{u_0^3}+4 \frac{m_x^2}{m^5}\right] \mathrm{d} x.
	\end{gather*}	
However, it is much more difficult to construct conservation laws for the derived integrable discrete analogues. As far as we know, conservation laws for the integrable semi-discrete CH equation \cite{Feng3} and semi-discrete DP equation \cite{dDP} have not been obtained. The reason may lie in the fact that the form of the semi-discrete equations are more complex than the continuous ones. As shown in Theorem~\ref{dmch}, we have to introduce two discrete analogues~$m_k$, $\tilde{m}_k$ corresponding to $m$. Though we express the semi-discrete analogue in an explicit form, and equation~\eqref{dmch-1} is nearly the form of conservation law, we fail to generate the conserved quantities for the semi-discrete mCH equation. How to construct the conserved quantities of the semi-discrete equations deserves further consideration.

\subsection*{Acknowledgements}
We greatly appreciate the anonymous referees' useful comments which help us improve the present paper significantly. G.-F.~Yu is supported by National Natural Science Foundation of China (Grant nos.~12175155, 12371251), Shanghai Frontier Research Institute for Modern Analysis and the Fundamental Research Funds for the Central Universities. B.-F.~Feng's work is supported by the U.S. Department of Defense (DoD), Air Force for Scientific Research (AFOSR) under grant No.~W911NF2010276.

\pdfbookmark[1]{References}{ref}
\LastPageEnding

\end{document}